\documentclass[11pt]{article}

\usepackage{amsmath}
\usepackage{amssymb}
\usepackage{amsthm}
\usepackage{todonotes}
\usepackage{tikz}
\usepackage{pgfplots}
\usepackage{graphicx}
\usepackage{epstopdf}
\usepackage[labelfont=bf,font=footnotesize]{caption}
\usepackage[labelfont=bf,font=footnotesize]{subcaption}
\usepackage{xcolor}
\usepackage{algorithmicx}
\usepackage{algpseudocode}
\usepackage{algorithm}
\usepackage{bbm}
\usepackage{makecell}
\usepackage{authblk}
\usepackage{fullpage,etoolbox}
\usepackage[colorlinks=true,citecolor=blue,linkcolor=blue,urlcolor=blue]{hyperref}
\usepackage[shortlabels]{enumitem}
\usepackage[round,sort,compress]{natbib}
\usepackage{listings}

\newtheorem{theorem}{Theorem}
\newtheorem{lemma}{Lemma}
\newtheorem{assumption}{Assumption}

\newtheorem{corollary}{Corollary} 
\newtheorem{remark}{Remark}
\newtheorem{definition}{Definition} 

\begin{document}
\title{Data-driven Distributed and Localized Model Predictive Control}
\author{Carmen~Amo~Alonso\thanks{Denotes equal contribution. $^1$C. Amo Alonso is with the Dept. of Computing and Mathematical Sciences, California Institute of Technology, Pasadena, CA. $^2$F. Yang is with the Dept. of Computer and Information Sciences, University of Pennsylvania, Philadelphia, PA. $^3$N. Matni is with the Dept. of Electrical and Systems Engineering, University of Pennsylvania, Philadelphia, PA. }$^{\ 1}$, Fengjun Yang$^{*2}$ and Nikolai Matni$^{3}$
\thanks{N. Matni is generously supported by NSF award CPS-2038873, NSF CAREER award ECCS-2045834, and Google Research Scholar Award.}
}
\maketitle

\begin{abstract}

Motivated by large-scale but computationally constrained settings, e.g., the Internet of Things, we present a novel data-driven distributed control algorithm that is synthesized directly from trajectory data. Our method, data-driven Distributed and Localized Model Predictive Control (D$^3$LMPC), builds upon the data-driven System Level Synthesis (SLS) framework, which allows one to parameterize \emph{closed-loop} system responses directly from collected open-loop trajectories. The resulting model-predictive controller can be implemented with distributed computation and only local information sharing. By imposing locality constraints on the system response, we show that the amount of data needed for our synthesis problem is independent of the size of the global system. Moreover, we show that our algorithm enjoys theoretical guarantees for recursive feasibility and asymptotic stability. Finally, we also demonstrate the optimality and scalability of our algorithm in a simulation experiment.

\end{abstract}


\section{Introduction}


Contemporary large-scale distributed systems such as the Internet of Things enjoy ubiquitous sensing and communication, but are locally resource constrained in terms of power consumption, memory, and computation power.  If such systems are to move from passive data-collecting networks to active distributed control systems, algorithmic approaches that exploit the aforementioned advantages subject to the underlying resource constraints of the network must be developed.  Motivated by this emerging control paradigm, we seek to devise a distributed control scheme that is (a) model-free, eliminating the need for expensive system identification algorithms, and (b) scalable in implementation and computation.  Our hypothesis is that for such systems, collecting local trajectory data from a small subset of neighboring systems is a far more feasible approach than deriving the intricate and detailed system models needed by model-based control algorithms. In this paper, we show that such a data-driven distributed control approach can scalably provide optimal performance and constraint satisfaction, along with feasibility and stability guarantees.

\textbf{Prior work:} 
The majority of data-driven control approaches have focused on providing solutions to the linear quadratic regulator (LQR) problem. Among these works, we focus on the \emph{direct methods} that do not require a system identification step \citep{hjalmarsson_iterative_1998,fazel_global_2018,mohammadi_linear_2020,bradtke_adaptive_1994,persis_formulas_2019,trentelman_informativity_2020}. Specifically, we highlight the work in \citet{persis_formulas_2019}, which applies behavioral systems theory to parametrize systems from past trajectories.\footnote{For a more in-depth treatment of behavioral system theory in the context of control problems, interested readers are referred to \citet{markovsky_behavioral_2021,dorfler_certainty_2021} and the references therein.} This idea has then given rise to several different data-enabled Model Predictive Control (MPC) approaches \citep{dorfler_bridging_2021,coulson_distributionally_2021,berberich_data_2020,xue_data_2021}. However, these approaches require gathering past trajectories of the global system, which hinders their scalability and challenges their applicability in the distributed setting. Even though some recent works have been developed where data-driven approaches were applied to the distributed setting, providing theoretical guarantees for these approaches remain difficult. For instance, the work in \citet{alemzadeh_d3pi_2021} provides an algorithm to solve the LQR problem where the dynamic matrices are unknown and communication only occurs at a local scale. This algorithm relies on the existence of ``auxiliary" links among subsystems, which can make its extension to an online approach (eg. MPC) very costly. It is also unclear how theoretical guarantees can be derived for an MPC approach relying on these techniques. On the other hand, some approaches use data-driven formulations to provide theoretical guarantees (recursive feasibility and asymptotic stability) in distributed MPC approaches where providing guarantees with conventional techniques is in general a hard problem and usually results in conservatism \citep{muntwiler_data_2020,sturz_distributed_2020}. However, in these cases knowledge of the system dynamics is assumed. It remains as an open question how to develop a \emph{scalable} distributed MPC approach where the system model is unknown and only \emph{local} measurements are available for each subsystem. It is important that such an approach enjoys the same theoretical guarantees of recursive feasibility and asymptotic stability as standard MPC approaches.

\textbf{Contributions:}  We address this gap and present a data-driven version of the model-based Distributed Localized MPC (DLMPC) algorithm for linear time-invariant systems in a noise-free setting. We rely on recent results on data-driven SLS \citep{xue_data_2021}, which show that optimization problems over system-responses can be posed using only libraries of past system trajectories without explicitly identifying a system model. We extend these results to the localized and distributed setting, where subsystems can only collect and communicate information within a local neighborhood. In this way, the model-based DLMPC problem can equivalently be posed using only \emph{local} libraries of past system trajectories, without explicitly identifying a system model. We then exploit this structure, together with the the separability properties of the objective function and constraints, and provide a distributed synthesis algorithm based on the Alternating Direction Method of the Multipliers (ADMM) \citep{boyd_distributed_2010} where only local information sharing is required. Hence, in the resulting implementation, each sub-controller solves a low-dimensional optimization problem defined over a local neighborhood, requiring only local data sharing and no system model. Since this problem is analogous to the model-based DLMPC problem \citep{amoalonso_distributed_2020,amoalonso_implementation_2021}, our approach directly inherits its guarantees for convergence, recursive feasibility and asymptotic stability \citep{amoalonso_theoretical_2021}. Through numerical experiments, we validate these results and further confirm that the complexity of the subproblems solved at each subsystem does not scale relative to the full size of the system.

\textbf{Paper structure:} In \S II we present the problem formulation.  In \S III, we introduce the necessary preliminaries on the Data-Driven SLS framework. \S IV expands on these results and provides a Data-Driven formulation of SLS that allows for locality constraints to be imposed such that only local information exchange is needed between subsystems. In \S V we apply these results to the Distributed and Localized MPC problem, and provide a distributed algorithm via ADMM that allows the MPC problem to be solved with only local information. In \S VI, we present a numerical experiment and we end in \S VII with conclusions and directions of future work.

{\textbf{Notation:}} 
Lower-case and upper-case Latin and Greek letters such as $x$ and $A$ denote vectors and matrices respectively, although lower-case letters might also be used for scalars or functions (the distinction will be apparent from the context). Bracketed indices denote time-step of the real system, i.e., the system input is  $u(t)$ at time $t$, not to be confused with $x_t$ which denotes the predicted state $x$ at time $t$. Superscripted variables in between curly brackets, e.g. $x^{\{k\}}$,  correspond to the value of $x$ at the $k^{th}$ iteration of a given algorithm. Calligraphic letters such as $\mathcal{S}$ denote sets, and lowercase script letters such as $\mathfrak{c}$ denotes a subset of $\mathbb{Z}^{+}$, i.e., $\mathfrak{c}=\left\{1,...,n\right\}\subset\mathbb{Z}^{+}$. Square bracket notation, i.e., $[x]_{i}$ denotes the components of $x$ corresponding to subsystem $i$ and, by extension $[x]_{j\in\mathcal S}$ denotes the components of $x$ corresponding to every subsystem $j
\in\mathcal S$. Boldface lower and upper case letters such as $\mathbf{x}$ and $\mathbf{K}$ denote finite horizon signals and block lower triangular (causal) operators, respectively:
\begin{equation*} 
\mathbf{x}=\left[\begin{array}{c} x_{0}\\x_{1}\\\vdots\\x_{T}\end{array}\right], ~
\mathbf K =   { {\scriptscriptstyle{\left[\begin{array}{cccc}K_{0}[0] & & & \\ K_{1}[1] & K_{1}[0] & & \\ \vdots & \ddots & \ddots & \\ K_{T}[T] & \dots & K_{T}[1] & K_{T}[0] \end{array}\right]}}},
\end{equation*}
where each $x_i$ is an $n$-dimensional vector, and each $K_{i}[j]$ is a matrix of compatible dimension representing the value of $K$ at the $j^\text{th}$ time-step computed at time $i$. $\mathbf{K}(\mathfrak{r},\mathfrak{c})$ denotes the submatrix of $\mathbf{K}$ composed of the rows and columns specified by $\mathfrak{r}$ and $\mathfrak{c}$ respectively. 

\section{Problem Formulation}

We consider a discrete-time linear time-invariant (LTI) system with dynamics
\begin{equation}\label{eqn:dynamics}
x(t+1) = Ax(t)+Bu(t)
\end{equation}
composed of $N$ interconnected subsystems. Each subsystem $i$ has dynamics
\begin{equation}\label{eqn:dynamics_i}
[x(t+1)]_i = \sum_{j\in\mathcal N_i} [A]_{ij}[x(t)]_i+[B]_{ii}[u(t)]_{i},
\end{equation}
where $[x]_i,[u]_i$ are suitable partitions of the global state $x\in\mathbb R^n$ and control input $u\in\mathbb R^p$ respectively. Similarly, $[A]_{ij}$ and $[B]_{ij}$ are the induced compatible block structure in the system matrices $(A,B)$, and the set $\mathcal N_i$ contains all subsystems $j$ such that $[A]_{ij}\neq0$. Note that we assume that $B$ is block-diagonal.

Our goal is to design a \emph{localized} MPC controller for the system when the system model $(A,B)$ is unknown. In this setup, each subsystem $i$ has access to a collection of past local state and input trajectories, and only local communication is possible among subsystems. To formalize the notion of locality, we model the interconnection topology of the system as a time-invariant unweighted directed graph $\mathcal G(\mathcal V, \mathcal E)$. In this graph, each subsystem $i$ is identified with a vertex $v_i \in \mathcal V$. An edge $e_{ij}\in\mathcal E$ exists whenever $j\in\mathcal N_i$. We say that a system is $d$-\textit{localized} if the sub-controllers are restricted to exchange their measurements and control inputs with neighbors at most $d$ hops away, as measured by the communication topology $\mathcal{G}(\mathcal V, \mathcal E)$. To formalize this idea, we first define the $d$-outgoing, $d$-incoming and $d$-external sets of a subsystem.
\begin{definition}{\label{def: in_out_set}}
For a graph $\mathcal{G}(\mathcal V, \mathcal E)$, the \textit{$d$-incoming set}, \textit{$d$-outgoing set} and \textit{$d$-external set} of subsystem $i$ are respectively defined as
\begin{itemize}
\item $\textbf{out}_{i}(d) := \left\{v_{j}\ |\  \textbf{dist}(v_{i} \rightarrow v_{j} ) \leq d\in\mathbb{N} \right\}$,
\item $\textbf{in}_{i}(d) \ \, := \left\{v_{j}\ |\ \textbf{dist}(v_{j} \rightarrow v_{i} ) \leq d\in\mathbb{N} \right\}$,
\item $\textbf{ext}_{i}(d) := \left\{v_{j}\ |\ \textbf{dist}(v_{j} \rightarrow v_{i} ) > d\in\mathbb{N} \right\}$,
\end{itemize}
where $\textbf{dist}(v_{j} \rightarrow v_{i} )$ denotes the distance between $v_{j}$ and $v_{i}$ i.e., the number of edges in the shortest path connecting $v_{j}$ to $v_{i}$.
\end{definition}
Hence, we can enforce a $d$-local information exchange constraint on the distributed MPC problem -- where the size of the local neighborhood $d$ is a \emph{design parameter} -- by imposing that each sub-controller policy $i$ can be computed using only states and control actions collected from $d$-hop incoming neighbors of subsystem $i$. Mathematically, we restrict the control action $[u_t]_i$ to be a function of the form
\begin{equation}
    [u_t]_i = \gamma^i_{t}\left([x_{0:t}]_{j\in \textbf{in}_i(d)},[u_{0:t-1}]_{j\in \textbf{in}_i(d)}\right),
    \label{eqn:comms}
\end{equation}
for all $t=0,\dots,T$ and $i=1,\dots,N$, where $\gamma_t^i$ is a measurable function of its arguments.  

With this in mind, we formulate the data-driven localized MPC problem. At each time-step, a finite-time constrained optimal control problem with horizon $T$ is solved, where the current state is used as the initial condition. Hence, at time step $\tau$ the MPC subroutine solves
\begin{align} \label{eqn:MPC}
&\hspace{-1mm} \underset{{x}_{t},u_{t}, \gamma_t}{\text{min}} &  &\sum_{t=0}^{T-1}f_{t}(x_{t},u_{t})+f_{T}(x_{T})\\
&\hspace{-1mm} \ \text{s.t.} &  
&x_{0} = x(\tau),\ x(t+1) = Ax(t)+Bu(t),\nonumber\\
&\hspace{-1mm} \ &     & x_{T}\in\mathcal{X}_{T},\, x_{t}\in\mathcal{X}_{t},\, u_{t}\in\mathcal{U}_{t},\ t\in[0,T-1], \nonumber\\
&\hspace{-1mm} \ &     & [u_t]_i = \gamma^i_{t}\left([x_{0:t}]_{j\in \textbf{in}_i(d)},[u_{0:t-1}]_{j\in \textbf{in}_i(d)}\right)\ \forall i,\nonumber
\end{align}
where the matrices $A$ and $B$ are unknown 
at the synthesis time, the $f_t(\cdot,\cdot)$ and $f_T(\cdot)$ are closed, proper, and convex cost functions, and $\mathcal{X}_t$ and $\mathcal{U}_t$ are closed and convex sets containing the origin. Moreover, in order to provide a local synthesis and implementation at each subsystem, we work under suitable structural assumptions between the cost function and state and input constraints.
\begin{assumption}{\label{assump: locality}}
In formulation \eqref{eqn:MPC} the objective function $f_{t}$  is such that $f_{t}(x,u)=\sum f_{t}^i([x]_{i},[u]_{i})$, and the constraint sets are such that $x\in\mathcal{X}_t=\mathcal{X}_{t}^1\times ... \times \mathcal{X}_{t}^N$, where $x \in \mathcal{X}_t$ if and only if $[x]_{i}\in\mathcal{X}_{t}^i$ for all $i$ and $t\in\{0,...,T\}$, and idem for $\mathcal{U}_t$.
\end{assumption}

In what follows we show problem \eqref{eqn:MPC} admits a distributed solution and implementation requiring only local data and no explicit estimate of the system model.

\section{Data-driven System Level Synthesis}\label{sec:sls}

In this section, we  introduce an abridged summary of some of the preliminary work on SLS \citep{wang_separable_2018,anderson_system_2019} and its extension to a data-driven formulation \citep{xue_data_2021} based on the behavioral framework in \citet{willems_introduction_1997}. In following sections, we build on these concepts to provide the necessary results to provide a tractable reformulation of problem \eqref{eqn:MPC}. 

\subsection{System Level Synthesis Parametrization}

The following is adapted from \S2 of \citet{anderson_system_2019}. Consider the dynamics of the LTI system \eqref{eqn:dynamics} evolving over a finite horizon $T$ and subject to an additive disturbance $w(t)$ at each time-step $t$. We can compactly express the dynamics as 
\begin{equation}\label{eqn:dynamics_compact}
\mathbf{x} = Z(\hat{A}\mathbf{x}+\hat{B}\mathbf{u})+\mathbf{w},
\end{equation}
where $\mathbf{x}$, $\mathbf{u}$ and $\mathbf{w}$ are the finite horizon signals corresponding to state, control input, and disturbance respectively. By convention, we define the disturbance to contain the initial condition, i.e., $\mathbf{w} = [x_{0}^\mathsf{T}\ w_{0}^\mathsf{T}\ \dots \ w_{T-1}^\mathsf{T}]^\mathsf{T}$. Here, $Z$ is the block-downshift matrix,\footnote{Matrix with identity matrices along its first block sub-diagonal and zeros elsewhere.} $\hat{A} := \text{blkdiag}(A, ..., A)$, and $\hat{B} := \text{blkdiag}(B, ..., B, 0)$. 

Let the control input of this system be a causal linear time-varying state-feedback controller, i.e., $u=\mathbf K x$ for controller $\mathbf K$ block-lower triangular. Then the closed-loop response of system \eqref{eqn:dynamics_compact} is given by:
\begin{subequations}\label{eqn:Phis}
\begin{equation}
\mathbf{x} = (I - Z(\hat{A} + \hat{B}\mathbf{K}))^{-1}\mathbf{w} =: \mathbf{\Phi}_x\mathbf{w} 
\end{equation}
\begin{equation}
\mathbf{u} = \mathbf{K}(I - Z(\hat{A} + \hat{B}\mathbf{K}))^{-1}\mathbf{w} =: \mathbf{\Phi}_u\mathbf{w}.
\end{equation}
\end{subequations}

The SLS approach relies on the so-called \emph{system responses} $\mathbf{\Phi}_x$ and $\mathbf{\Phi}_u$ to parametrize the set of achievable closed-loop behaviors of the system. This can be done by virtue of the following theorem:
\begin{theorem}{\label{thm:SLS}} \emph{(Theorem 2.1 of \citet{anderson_system_2019})}
For the system  \eqref{eqn:dynamics} evolving under the state-feedback policy $\mathbf u = \mathbf K \mathbf x$ with $\mathbf{K}$ a block-lower-triangular matrix, the following are true
\begin{enumerate}
    \item The affine subspace
    \begin{equation}\label{eqn:Z_AB}
        Z_{AB}\mathbf{\Phi}:=\left[I-Z\hat A\ \ -Z\hat B\right]\left[\begin{array}{c}\mathbf{\Phi}_{x}\\\mathbf{\Phi}_{u}\end{array}\right] = I
    \end{equation}
    with block-lower-triangular $\{\mathbf{\Phi}_{x},\mathbf{\Phi}_{u}\}$ parameterizes all possible system responses \eqref{eqn:Phis}.
    
    \item For any block lower-triangular matrices $\left\{\mathbf{\Phi}_{x},\mathbf{\Phi}_{u}\right\}$ satisfying \eqref{eqn:Z_AB}, the controller $\mathbf{K} = \mathbf{\Phi}_{u}\mathbf{\Phi}_{x}^{-1}$ achieves the desired response \eqref{eqn:Phis} from $\mathbf w \mapsto (\mathbf x,\mathbf u)$.
\end{enumerate}
\end{theorem}

Theorem \ref{thm:SLS} allows one to reformulate an optimal control problem over state and input variables $(\mathbf x, \mathbf u)$ as an equivalent one over system responses $(\mathbf{\Phi}_x,\mathbf{\Phi}_u)$. A detailed description of how to do this for several standard control problems is provided in §2 of \citet{anderson_system_2019}. 

\subsection{Locality Constraints in System Level Synthesis}

One of the advantages of the SLS framework is that it allows one to take into account the structure of a LTI system \eqref{eqn:dynamics} composed of subsystems \eqref{eqn:dynamics_i} \citep{wang_separable_2018}. In particular, the information constraints defined in equation \eqref{eqn:comms} can be enforced via \emph{locality constraints} on the system responses \eqref{eqn:Phis}. To see this, note that according to Theorem \ref{thm:SLS}, the controller $\mathbf K=\mathbf \Phi_u\mathbf\Phi_x^{-1}$ achieves the desired closed-loop behavior characterized by $\mathbf\Phi_x$ and $\mathbf\Phi_u$. Such a controller can be implemented as:
\begin{equation}\label{eqn:implementation}
   \mathbf{u}=\mathbf{\Phi}_u\mathbf{\hat w},\quad \mathbf{x}=(\mathbf{\Phi}_x-I)\mathbf{\hat w},\quad
   \mathbf{\hat w}=\mathbf{ x}-\mathbf{\hat x},
\end{equation}
where $\mathbf{\hat x}$ can be interpreted as a nominal state trajectory, and $\mathbf{\hat w} = Z\mathbf w$ is a delayed reconstruction of
the disturbance. Given this controller implementation, any structure imposed on the maps $(\mathbf\Phi_x,\mathbf\Phi_u)$ is directly translated to structure on controller implementation \eqref{eqn:implementation}. Hence, one can transparently impose information exchange constraints as suitable sparsity structure on the responses $(\mathbf\Phi_x,\mathbf\Phi_u)$. We formalize this idea as follows.

\begin{definition}\label{def:locality}
Let $[\mathbf{\Phi}_{x}]_{ij}$ be the submatrix of system response $\mathbf{\Phi}_x$ describing the map from disturbance $[w]_{j}$ to the state $[x]_i$ of subsystem $i$. The map $\mathbf{\Phi}_{x}$ is \textit{d-localized} if and only if for every subsystem $j$, $[\mathbf{\Phi}_{x}]_{ij}=0\ \forall\ i\not\in\mathbf{out}_{j}(d)$. The definition for \textit{d-localized} $\mathbf{\Phi}_u$ is analogous but with disturbance to control action $[u]_i$.
\end{definition}

By simply enforcing the system responses $\mathbf {\Phi}_x$ and $\mathbf {\Phi}_u$ to have a $d$-localized and $(d+1)$-localized sparsity pattern,\footnote{We impose $\mathbf{\Phi}_{u}$ to be $(d+1)$-localized because the ``boundary" controllers at distance $d+1$ must take action in order to localize the effects of a disturbance within the region of size $d$,  (for more details the reader is referred to \citet{anderson_system_2019}).} only a local subset $[{\mathbf{\hat w}}]_{\mathbf{in}_i(d)}$ of ${\mathbf{\hat w}}$ is necessary for subsystem $i$ to construct for $[\mathbf x]_i$ and $[\mathbf u]_i$. Therefore, the information exchange needed between subsystems during the controller synthesis is limited to $d$-hop incoming and outgoing neighbors, as defined by constraint  \eqref{eqn:comms}. Moreover, this constraint can be imposed on the system responses as an affine subspace.
\begin{definition} 
A subspace $\mathcal L_d$ enforces a $d$-locality constraint if $\mathbf \Phi_x,\mathbf\Phi_u\in\mathcal L_d$ implies that $\mathbf \Phi_x$ is $d$-localized and $\mathbf \Phi_u$ is $(d + 1)$-localized. A system $(A, B)$ is then $d$-localizable if the intersection of $\mathcal L_d$ with the affine space of achievable system responses \eqref{eqn:Z_AB} is non-empty.
\end{definition}

As introduced in \citet{amoalonso_distributed_2020}, when local model information i.e. $[A]_{ij},\ [B]_i$ is available, locality constraints allow the MPC subroutine \eqref{eqn:MPC} to be reformulated into a tractable problem that admits a distributed solution. Because the setting considered in this paper has no driving noise, only the first block columns of $\mathbf \Phi_x$ and $\mathbf \Phi_u$ are necessary to characterize system behavior. 
For this reason, in what follows we abuse notation and use $\mathbf \Phi_x$ and $\mathbf \Phi_u$ to denote only the first block column of the original block-lower triangular matrices. Therefore, the change of variables defined in equation \eqref{eqn:Phis} reduces to $\mathbf x = \mathbf \Phi_x x_0$ and $\mathbf u = \mathbf \Phi_u x_0$.

Hence, in a \emph{model-based} setting the MPC subroutine \eqref{eqn:MPC} can be reformulated as
\begin{equation} \label{eqn:DLMPC}
\begin{aligned}
& \underset{\mathbf \Phi}{\text{minimize}}
& &f(\mathbf \Phi x_0)\\
& \ \text{s.t.} &  &\begin{aligned} 
     & x_0 = x(t),\ \mathbf \Phi x_{0}\in\mathcal{P}, \\
     & Z_{AB}\mathbf \Phi=I, \ \mathbf{\Phi}\in\mathcal L_d,
\end{aligned}
\end{aligned}
\end{equation}
where $f$ is a convex function compatible with all the $f_t$ in formulation \eqref{eqn:MPC}, and $\mathcal P$ is defined so that $\mathbf \Phi x_{0}\in\mathcal{P}$ if and only if $x(t)\in\mathcal X_t$ and $u(t)\in\mathcal U_t$ for all $t=0,\dots,T$. For additional details on this formulation and its distributed synthesis and implementation readers are referred to \citet{amoalonso_implementation_2021} and references therein.

 \begin{remark} Although $\mathcal L_d$ is always a convex subspace, not all systems are $d$-localizable. For systems that are not $d$-localizable, constraint $\mathbf{\Phi}\in\mathcal L_d$ would lead to an infeasible subroutine \eqref{eqn:DLMPC}. The locality diameter $d$ can be viewed as a design parameter that is application dependent. For the remainder of the paper, we assume that there exists a $d\ll n$ such that the system $(A,B)$ to be controlled is $d$-localizable. Notice that the parameter $d$ is tuned independently of the horizon $T$, and captures how ``far'' in the interconnection topology a disturbance striking a subsystem is allowed to spread.
 \end{remark}

\subsection{Data-driven System Level Synthesis}
This subsection is adapted from \S 2 of \citet{xue_data_2021}. 

Behavioral system theory \citep{willems_introduction_1997,persis_formulas_2019,markovsky_behavioral_2021,dorfler_certainty_2021} offers a natural way of studying the behavior of a dynamical system in terms of its input/output signals. In particular, Willem's Fundamental Lemma \citep{willems_introduction_1997} offers a parametrization of state and input trajectories based on past trajectories as long as the data matrix satisfies a notion of persistance of excitation.
\begin{definition}
A finite-horizon signal $\mathbf x$ with horizon $T$ is \emph{persistently exciting} (PE) of order $L$ if the Hankel matrix 
$$H_L(\mathbf x):= { {\scriptscriptstyle{\left[\begin{array}{cccc}x(0) & x(1) & \dots & x(T-L) \\ x(1) & x(2) & \dots & x(T-L+1)\\ \vdots & \vdots & \ddots & \vdots \\ x(L-1) & x(L) & \dots & x(T-1) \end{array}\right]}}}$$ 
has full rank.
\end{definition}

\begin{lemma}\label{lemm:Willem}
(Willem's Fundamental Lemma \citep{willems_introduction_1997}) Consider the LTI system \eqref{eqn:dynamics} with controllable $(A, B)$ matrices, and assume that there is no driving noise. Let $\{\tilde{\mathbf x},\tilde{\mathbf u}\}$ be the state and input signals generated by the system over a horizon $T$. If $\tilde{\mathbf u}$ is PE of order $n+L$, then the signals $\mathbf x$ and $\mathbf u$ are valid trajectories of length $L$ of the system \eqref{eqn:dynamics} if and only if
\begin{equation}\label{eqn:behavioral}
\begin{bmatrix}
     \mathbf x \\ \mathbf{u}
\end{bmatrix}
=
 H_L(\tilde{\mathbf{x}},\tilde{\mathbf u}) \
 g ~~\text{for some } g\in\mathbb{R}^{T-L+1},
\end{equation}
where $H_L(\tilde{\mathbf{x}},\tilde{\mathbf u}) := \begin{bmatrix}
 H_L(\tilde{\mathbf{x}})^\intercal & H_L(\tilde{\mathbf{u}})^\intercal
 \end{bmatrix}^\intercal $.
\end{lemma}
\vspace{2mm}

A natural connection can be established between the data-driven parametrization \eqref{eqn:behavioral} and the SLS parametrization \eqref{eqn:Phis}. In particular, the achievability constraint \eqref{eqn:Z_AB} can be replaced by a data-driven representation by applying Willems' Fundamental Lemma \citep{willems_introduction_1997} to the columns of the system responses. Given a system response, we denote the set of columns corresponding to subsystem $i$ as $\mathbf \Phi^{i}$, i.e., 
\[
\mathbf \Phi = \begin{bmatrix} \mathbf \Phi^{1} & \mathbf \Phi^{2} 
& ... & \mathbf \Phi^{N}\end{bmatrix}.
\]
The key insight is that, by definition of the system responses \eqref{eqn:Phis}, $\mathbf \Phi_x^i$ and $\mathbf \Phi_u^i$ are the impulse response of $\mathbf x$ and $\mathbf u$ to $[x_0]_i$, which are themselves, valid system trajectories that can be characterized using Willems’ Fundamental Lemma. This can be seen from the following decomposition of the trajectories
$$
\begin{bmatrix} \mathbf x \\ \mathbf u \end{bmatrix} = \mathbf \Phi [x_0]_i
= \sum_{i=1}^N \mathbf \Phi^{i} [x_0]_i.
$$

\begin{lemma}\label{lemm:Anton}
(Lemma 1 of \citet{xue_data_2021}) Given the assumptions of Lemma \ref{lemm:Willem}, the set of feasible solutions to constraint \eqref{eqn:Z_AB} over a time horizon $t = 0, 1, . . . , L-1$ can be equivalently characterized as:
\begin{equation}\label{eqn:behavioral_sls}
H_L(\tilde{\mathbf{x}},\tilde{\mathbf{u}}) \mathbf{G},\text{  for all $\mathbf G$ s.t. $H_1(\tilde{\mathbf{x}})\mathbf G = I$.}
\end{equation}
\end{lemma}
\vspace{2mm}


The following Corollary follows naturally from Lemma \ref{lemm:Anton}, and will be useful later to provide locality results in this data-driven parametrization. 
\begin{corollary} \label{cor:Anton}
The following is true:
\begin{equation*}
\{\mathbf{\Phi}^{i}:\ Z_{AB}\mathbf\Phi^i = I^i\} = 
\{ H_L(\tilde{\mathbf{x}},\tilde{\mathbf{u}}) \mathbf G^i: \mathbf G^i \text{ s.t. } H_1(\tilde{\mathbf x})\mathbf G^i = I^{i}\},
\end{equation*}
where $I^{i}$ denotes the $i$-th block column of the identity matrix.
\end{corollary}
\begin{proof}
This follows directly from Lemma \ref{lemm:Anton} by noting that the constraints can be separated column-wise.
\end{proof}

While this connection between SLS and the behavioral formulation does not offer an immediate benefit, we will build on it in the following sections to equip the data-driven parametrization \eqref{eqn:behavioral_sls} with locality constraints so as to provide a reformulation of the localized MPC subroutine \eqref{eqn:MPC}.

\section{Localized Data-Driven System Level Synthesis}\label{dd-l-SLS}
In this section we present the necessary results that allow us to recast the constraints in \eqref{eqn:DLMPC} in a localized data-driven parametrization. We first provide a naive parametrization of system responses subject to locality constraints based on Lemma \ref{lemm:Anton} in terms of $\mathbf G$. We then build on this parameterization and show that localized system responses can be characterized using only locally collected trajectories.

\subsection{Locality Constraints in Data-Driven System Level Synthesis}
We start by rewriting the locality constraints using the data-driven parameterization \eqref{eqn:behavioral_sls}.

\begin{lemma}\label{lemm:localized_G}
Consider the LTI system \eqref{eqn:dynamics} with controllable $(A,B)$ matrices, where each subsystem \eqref{eqn:dynamics_i} is subject to locality constraints \eqref{eqn:comms}. Assume that there is no driving noise. Given the state and input trajectories $\{\tilde{\mathbf x},\tilde{\mathbf u}\}$ generated by the system over a horizon $T$ with $\mathbf u$ PE of order at least $n+L$, the following parametrization over $\mathbf G$ characterizes all possible $d$-localized system responses over a time span of $L - 1$:
\begin{align}\label{eqn:localized_G}
 H_L(\tilde{\mathbf{x}},\tilde{\mathbf{u}})
\mathbf G,~ & \text{for all}~ \mathbf G~\text{s.t.}\ H_1(\tilde{\mathbf x})\mathbf G=I, \\
  & H_L([\tilde{\mathbf{x}}]_i)\mathbf G^j=0\ \forall j\notin\textbf{in}_i(d),\ \nonumber \\
  & H_L([\tilde{\mathbf{u}}]_i)\mathbf G^k=0\ \forall k\notin\textbf{in}_i(d+1), \nonumber
 \end{align}
 for all $i= 1,\dots,N.$
\end{lemma}

\begin{proof}
We aim to show that $$\{\mathbf{\Phi}:\ Z_{AB}\mathbf\Phi = I,\mathbf\Phi\in\mathcal L_d\} = \{  H_L(\tilde{\mathbf{x}},\tilde{\mathbf{u}})  \mathbf G: \mathbf G \text{ s.t. } \eqref{eqn:localized_G}\}.$$

$(\subseteq)$ First, suppose that $\mathbf{\Phi} \in \mathcal L_d$ satisfies that $Z_{AB}\mathbf{\Phi} = I$. From Lemma \ref{lemm:Anton}, we immediately have that there exists a matrix $\mathbf G$ s.t. $\mathbf \Phi =  H_L(\tilde{\mathbf{x}},\tilde{\mathbf{u}}) \mathbf G.$
Thus, we need only verify that this $\mathbf G$ satisfies the linear constraint in \eqref{eqn:localized_G}. This follows directly from the assumption that $\mathbf \Phi \in \mathcal L_d$, which states that
\begin{align*}
  & H_L([\tilde{\mathbf{x}}]_i)\mathbf G^j=[\mathbf \Phi_x]_{ij}=0\ \forall j\notin\textbf{in}_i(d),\ \nonumber \\
  & H_L([\tilde{\mathbf{u}}]_i)\mathbf G^k=[\mathbf\Phi_u]_{ik}=0\ \forall k\notin\textbf{in}_i(d+1).
\end{align*}
Hence, $\mathbf\Phi \in \text{RHS}$, proving this direction.

$(\supseteq)$ Now suppose that there exists a $\mathbf G$ that satisfies the constraints on the RHS and let $\mathbf \Phi = H_L(\tilde{\mathbf{x}},\tilde{\mathbf{u}}) \mathbf G.$ Since $H_1(\tilde{\mathbf x})\mathbf{G}=I$, from Lemma 2, we have that $\mathbf \Phi$ is achievable. From the other two constraints, we have that $\mathbf \Phi \in \mathcal L_d$, proving this direction and hence the lemma.
\end{proof}

It is important to note that even though Lemma \ref{lemm:localized_G} allows one to capture the locality constraint \eqref{eqn:comms} by simply translating the locality constraints over $\mathbf\Phi$ to constraints over $\mathbf G$, it cannot be implemented with only local information exchange. In order to satisfy the constraints \eqref{eqn:localized_G}, each subsystem has to have access to global state and input trajectories and construct a global Hankel matrix. The PE condition of Lemma \ref{lemm:Willem} further implies that the length of the trajectory that needs to be collected grows with the dimension of the global system state. In what follows we show how constraint \eqref{eqn:localized_G} can further be relaxed to only require local information without introducing any additional conservatism.

\subsection{Localized Data-driven System Level Synthesis}

In this subsection we show that constraint \eqref{eqn:localized_G} can be enforced (i) with local communication between neighbors i.e. no constraints are imposed outside each subsystem $d$-neighborhood, and (ii) the amount of data needed i.e., trajectory length, only scales with the size of the $d$-localized neighborhood, and not the global system. We start by providing a result that allows constraint \eqref{eqn:localized_G} to be satisfied with local information only.

\begin{definition}
Given a subsystem $i$ satisfying the local dynamics \eqref{eqn:dynamics_i}, we define its \emph{augmented $d$-localized subsystem} as the system composed by the states $[x]_{\textbf{in}_i(d+1)}$ and augmented control actions $[\bar{u}]_i := ([u]_{\textbf{in}_i(d+2)}^\intercal\;[x]_j^\intercal)^\intercal,\;\forall j\ s.t.\ \textbf{dist}(j\to i) = d+2$. That is, the system given by
\begin{equation}\label{eqn:agumented_subsytem}
[x(t+1)]_{\textbf{in}_i(d+1)}=[A]_{\textbf{in}_i(d+1)}[x(t)]_{\textbf{in}_i(d+1)}
+[\bar{B}]_{\textbf{in}_i(d+1)}[\bar{u}(t)]_i,
\end{equation}
with $\bar{B}:=\begin{bmatrix}
[B]_{\textbf{in}_i(d+2)} & [A]_{ij} 
\end{bmatrix}$ $\forall\ j$ s.t. $dist(j\rightarrow i) = d+2$.
\end{definition}

Notice that by treating the state of the boundary subsystems as additional control inputs, we can view the augmented $d$-localized system as a standalone LTI system.

\begin{lemma}\label{lemm:boundary}
For $i=1,\dots,N$, let $\mathbf\Psi^i$ be an achievable system response for the augmented $d$-localized subsystem \eqref{eqn:agumented_subsytem} of subsystem $i$. Further assume that each $\mathbf\Psi^i$ satisfies constraints \eqref{eqn:augmented_constraints}:
\begin{subequations}\label{eqn:augmented_constraints}
\begin{equation}
    [\mathbf \Psi^{i}_x]_{j} = 0,\ \forall j\text{ s.t. } d+1\leq \textbf{dist}(j\rightarrow i)\leq d+2,
\end{equation}
\begin{equation}
    [\mathbf \Psi^{i}_u]_{j} = 0,\ \forall j\text{ s.t. } \textbf{dist}(j\rightarrow i)= d+2
\end{equation}
\end{subequations}
for all $i$. Then, the system response $\mathbf \Phi$ defined by \eqref{eqn:aumgmented_Phi} is achievable for system \eqref{eqn:dynamics} and $d$-localized. 
\begin{equation}\label{eqn:aumgmented_Phi}
[\mathbf \Phi]_{ij} := \begin{cases}
    [\mathbf \Psi^{i}]_{j}, & \forall j\in\textbf{in}_i(d+1) \\
    0, & \text{otherwise}
    \end{cases}
\end{equation}
for all $i=1,\dots,N$ is also achievable and $d$-localized.
\end{lemma}

\begin{proof}
First, from the fact that $\mathbf \Psi^{i}$ is achievable for all $i=1,\dots,N$, we have that $ \Phi_x[0] = I$ by construction. Thus, to show that $\mathbf \Phi$ is achievable, it suffices to show that
\begin{equation*}
    \Phi_x[t+1] = A \Phi_x[t] + B \Phi_u[t], \quad\forall\;0 \leq t \leq T-1.
\end{equation*}
We show this block-column-wise. Specifically, we show that the block columns $\Phi_{x}^{i}$ and $\Phi_{u}^{i}$ associated with each subsystem satisfy
\begin{equation}\label{eq:column_achievability}
        \Phi_{x}^{i}[t+1] = A \Phi_{x}^{i}[t] + B \Phi_{u}^{i}[t], \quad\forall\;0 \leq t \leq T-1.
\end{equation}

We further partition the rows of these block-columns into four subsets as follows:
$$ \Phi_x^i = \begin{bmatrix}
[\Phi_{x}^{i}]_{\textbf{in}_i(d)}^\intercal &
[\Phi_{x}^{i}]_{d+1}^\intercal &
[\Phi_{x}^{i}]_{d+2}^\intercal &
[\Phi_{x}^{i}]_{\textbf{ext}_i(d+2)}^\intercal
\end{bmatrix}^\intercal,
$$
where the notation $[\Phi_{x}^i]_{k}$ represents the entries of $\Phi_x$ corresponding to subsystems $k$-hops away from the $i$-th subsystem. Identical notation holds for the partition of $\Phi_u^i$. 

Using this partition, we have the following for $\Phi_x^i$ and $\Phi_u^i[t]$ given their definition in terms of $\mathbf \Psi$:
\begin{equation*}
    \Phi_{x}^{i}[t] = \begin{bmatrix}
[\Psi_{x}^{i}[t]]_{\textbf{in}_i(d)}\\ 0 \\ 0 \\ 0
\end{bmatrix}, \quad
\Phi_{u}^{i}[t] =  \begin{bmatrix}[\Psi_{u}^{i}[t]]_{\textbf{in}_i(d)}\\
[\Psi_{u}^{i}[t]]_{d+1}\\ 0 \\ 0
\end{bmatrix}.
\end{equation*}
We also partition the dynamics matrices $A$ and $B$ accordingly, where
\begin{gather*}
A = \begin{bmatrix}
A^{\textbf{in}(d)}_{\textbf{in}(d)} & A^{\textbf{in}(d)}_{d+1} & 0 & 0\\
A^{d+1}_{\textbf{in}(d)} & A^{d+1}_{d+1} & A^{d+1}_{d+2} & 0\\
0 & A^{d+2}_{d+1} & A^{d+2}_{d+2} & A^{d+2}_{\textbf{ext}_i(d+2)}\\
0 & 0 & A^{\textbf{ext}_i(d+2)}_{d+2} & A^{\textbf{ext}_i(d+2)}_{\textbf{ext}_i(d+2)}
\end{bmatrix}, \\
B = \begin{bmatrix}
B^{\textbf{in}(d)}_{\textbf{in}(d)} & 0 & 0 & 0\\
0 & B^{d+1}_{d+1} & 0 & 0\\
0 & 0 & B^{d+2}_{d+2} & 0\\
0 & 0 & 0 & B^{\textbf{ext}_i(d+2)}_{\textbf{ext}_i(d+2)}
\end{bmatrix}.
\end{gather*}
Here, the superscript represents an index on the block row and the subscript represents an index on the block column. The sparsity pattern of the partition follows directly from the definition of augmented $d$-localized subsystems and the subsystem dynamics \eqref{eqn:dynamics_i}.

We can now show that equation \eqref{eq:column_achievability} holds for each of the $i^\text{th}$ block-columns and $\Phi^{i}$ is an achievable impulse response of the system. First, note that
\begin{align*}
    [\Phi_{x}^{i}[t+1]]_{\mathbf{in}_i(d)} &= [A\Phi_{x}^{i}[t] + B \Phi_{u}^{i}[t]]_{\textbf{in}_i(d)} \\
    &= A^{\textbf{in}_i(d)}_{\textbf{in}_i(d)} [\Psi_{x}^{i}[t]]_{\textbf{in}_i(d)} + 
        B^{\textbf{in}_i(d)}_{\textbf{in}_i(d)} [\Psi_{u}^{i}[t]]_{\textbf{in}_i(d)} + B^{\textbf{in}_i(d)}_{d+1}[\Psi_{u}^{i}[t]]_{d+1} \\
    &= [\Psi_{x}^{i}[t+1]]_{\mathbf{in}_i(d)},
\end{align*}
where the second equality comes from the sparsity patterns of $A, B$, and $\mathbf \Phi^i$, and the third equality from the achievability of $\mathbf \Psi^{i}$. Similarly, to show that the boundary subsystems satisfy the dynamics, we note that
\begin{align*}
    [\Phi_{x}^{i}[t+1]]_{d+1} &= A^{d+1}_{\mathbf{in}_i(d)} [\Psi_{x}^{(i)}(t)]_{\mathbf{in}_i(d)} + B^{d+1}_{d+1}[\Psi_{u}^{(i)}(t)]_{d+1} \\
    &= [\Psi_{x}^{(i)}(t+1)]_{d+1}\\
    &= 0.
\end{align*}
Lastly, from the sparsity pattern of the dynamic matrices and $\Phi^i[t]$, we trivially have that
\[
[\Phi_{x}^{(i)}(t+1)]_{\textbf{ext}_i(d)} = [A\Phi_{x}^{i}[t] + B \Phi_{u}^{i}[t]]_{\textbf{ext}_i(d)} = 0,
\]
concluding the proof for the achievability of $\mathbf \Phi$. We end by noting that $\mathbf \Phi$ is $d$-localized by construction.
\end{proof}

In light of this result, locality constraints as in Definition \ref{def:locality} i.e., $[\mathbf{\Phi}_{x}]_{ij}=0\ \forall\ i\not\in\mathbf{out}_{j}(d)$, do not need to be imposed on every subsystem $i\not\in\mathbf{out}_{j}(d)$. Instead, it suffices to impose this constraint only on subsystems $i$ at a distance $d+2$ of subsystem $j$. Intuitively, this can be seen as a constraint on the propagation of a signal: if $[w]_j$ has no effect on subsystem $i$ at distance $d+1$ because $[\mathbf{\Phi}_{x}]_{ij}=0$, then the propagation of that signal is stopped and localized within that neighborhood. This idea will allow us to reformulate constraint \eqref{eqn:localized_G} so that it can be imposed with only local communications.

However, despite the fact that locality constraints can now be achieved with local information exchanges, the amount of data that needs to be collected scales with the global size of the network $n$ because we require that the control trajectory be at least PE of order at least $n+L$. In the following theorem, we build upon the previous results and show how this requirement can also be reduced to only depend on the size of a $d$-localized neighborhood. 

\begin{theorem}\label{thm:localized-dd-sls}
Consider the LTI system \eqref{eqn:dynamics} composed of subsystems \eqref{eqn:dynamics_i}, each with controllable $([A]_{\textbf{in}_i(d+2)},[B]_{\textbf{in}_i(d+2)})$ matrices for the augmented $d$-localized subsystem $i$. Assume that there is no driving noise and that the local control trajectory at the $d$-localized subsystem $[\tilde{\mathbf u}]_{\textbf{in}_i(d+1)}$ is PE of order at least $n_{\textbf{in}_i(d)}+L$, where $n_{\textbf{in}_i(d)}$ is the dimension of $[\tilde{\mathbf x}]_{\textbf{in}_i(d)}$. Then, $\mathbf \Phi$ is an achievable $d$-localized system response for each subsystem \eqref{eqn:dynamics_i} if and only if it can be written as
\begin{subequations}\label{eqn:thm_Phi}
\begin{equation} \label{eqn:thm_Phi1}
    [\mathbf \Phi^{i}]_{\textbf{in}_i(d)} = H_L([\tilde{\mathbf x}]_{\textbf{in}_i(d+1)},[\tilde{\mathbf u}]_{\textbf{in}_i(d+1)}) \mathbf G^i, 
\end{equation}
\begin{equation}\label{eqn:thm_Phi2}
    [\mathbf \Phi^{i}]_{\textbf{ext}_i(d+1)} = 0,
\end{equation}
\end{subequations}
where $\mathbf G^i$ satisfies
\begin{subequations}\label{eqn:thm}
\begin{equation}
\label{eqn:thm-achievability}
    H_1([\tilde{\mathbf x}]_{\textbf{in}_i(d+1)})\mathbf G^i =I^{i}, \\
\end{equation}
\begin{equation} 
\label{eqn:thm-locality1}
    H_L([\tilde{\mathbf{x}}]_j)\mathbf G^i=0 \ \forall i, j\text{ s.t. } d+1 \leq \textbf{dist}(j\rightarrow i) \leq d+2,\\
\end{equation}
\begin{equation}
\label{eqn:thm-locality2}
  H_L([\tilde{\mathbf{u}}]_j)\mathbf G^i=0\ \forall i,j\text{ s.t. } \textbf{dist}(j\rightarrow i) = d+2.
\end{equation}
\end{subequations}
  
\end{theorem}

\vspace{1mm} \begin{proof}
($\Rightarrow$) We first show that all $d$-localized system responses $\mathbf \Phi$ can be parameterized by a corresponding set of matrices $\{\mathbf G^{i}\}_{i=1}^N$. First, we note that since $\mathbf \Phi$ is $d$-localized, each $d$-localized subsystem impulse response $[\mathbf \Phi^{i}]_{\textbf{in}_i(d+1)}$ is achievable on the augmented $d$-localized subsystem $i$. Thus, from applying Corollary \ref{cor:Anton}, we have that there exists $\mathbf G^i$ satisfying constraint \eqref{eqn:thm-achievability} such that 
\[
[\mathbf \Phi^{i}]_{\textbf{in}_i(d)} = H_L([\tilde{\mathbf x}]_{\textbf{in}_i(d)}, [\tilde{\mathbf u}]_{\textbf{in}_i(d+1)}) \mathbf G^i.
\]
Since $\mathbf \Phi$ is $d$-localized, we have that $\mathbf G^i$ satisfies both constraints \eqref{eqn:thm-locality1} and \eqref{eqn:thm-locality2}, concluding the proof in this direction.\\
($\Leftarrow$) Now we show that if each $\mathbf G^i$ satisfies the constraint \eqref{eqn:thm} for all $i=1,\dots,N$, then the resulting $\mathbf \Phi$ is achievable and localized. Consider the augmented $d$-localized subsystem $i$ and define 
\[
\mathbf \Psi^{i} =  H_L([\tilde{\mathbf x}]_{\textbf{in}_i(d)},[\tilde{\mathbf u}]_{\textbf{in}_i(d+1)}) \mathbf G^i.
\]
From Corollary \ref{cor:Anton}, we have that $\mathbf \Psi^{i}$ is an achievable impulse response on the augmented $d$-localized subsystem $i$. Moreover, by construction it satisfies the sparsity condition in Lemma \ref{lemm:boundary}. Thus, constructing $\mathbf \Phi$ using $\mathbf \Psi$ as in equation \eqref{eqn:aumgmented_Phi} we conclude that $\mathbf \Phi$ is an achievable and $d$-localized system response.
\end{proof}

\begin{corollary}
Consider a function $g: \mathbf \Phi \to \mathbb{R}$ such that
\[g(\mathbf\Phi) = \sum_{i=1}^{N}g^i([\mathbf \Phi^{i}]_{\textbf{in}_i(d+1)}).\]
Then, solving the optimization problem
\[ \min\ g(\mathbf \Phi)\ s.t.\ Z_{AB}\mathbf \Phi = I,\ \mathbf \Phi \in \mathcal L_d\]
is equivalent to solving
\[ \min\ \sum_{i=1}^{N}g^i\Big(H_L\big([\tilde{\mathbf x}]_{\textbf{in}_i(d)}, [\tilde{\mathbf u}]_{\textbf{in}_i(d+1)}\big)\mathbf G^i\Big)\]
\[s.t.\ \mathbf G^i\ \text{satisfies}\ \eqref{eqn:thm}\text{ for all } i=1\dots,N,\]
and then constructing the $d$-localized system response $\mathbf{\Phi}$ as per equation \eqref{eqn:thm_Phi}.
\end{corollary}

This result provides a data-driven approach in which locality constraints, as in equation \eqref{eqn:comms}, can be seamlessly considered and imposed by means of an affine subspace where only local information exchanges are necessary. Moreover, the amount of data needed to parametrize the behavior of the system does not scale with the size of the network but rather with $d$, the size of the localized region, which is usually much smaller than $n$. To the best of our knowledge, this is the first such result. As we show next, this will prove key in extending data-driven SLS to the distributed setting.

\section{Distributed AND Localized algorithm for Data-driven MPC}\label{sec:admm}

In this section we make use of the results on localized data-driven SLS from previous sections and apply them to reformulate the MPC subproblem \eqref{eqn:MPC}. We provide a distributed and localized algorithmic solution that does not scale with the size of the network. Lastly, we comment on the theoretical guarantees of this data-driven DLMPC (D$^3$LMPC) approach in terms of convergence, recursive feasibility and asymptotic stability.

\subsection{System Level Synthesis reformulation of data-driven MPC}
In light of Theorem \ref{thm:localized-dd-sls}, we can write the MPC subproblem \eqref{eqn:MPC} in terms of the variable $\mathbf G$ and localized Hankel matrices $H_L([\tilde{\mathbf x}]_{\mathbf{in}_i(d+2)}, [\tilde{\mathbf u}]_{{\mathbf{in}_i(d+2)}})$. To do this, we proceed as in reformulation \eqref{eqn:DLMPC} and rely on the equivalence between standard and data-driven SLS parametrizations, i.e., $\mathbf \Phi =  H_L(\tilde{\mathbf{x}},\tilde{\mathbf{u}}) \mathbf G \Leftrightarrow Z_{AB}\mathbf \Phi = I$. We make use of Lemma \ref{thm:localized-dd-sls} to recast the locality constraints \eqref{eqn:comms} into local affine constraints. Hence, we rewrite problem \eqref{eqn:MPC} as 
\begin{equation} \label{eqn:dd-DLMPC}
\begin{aligned}
& \underset{\mathbf \Phi, \{\mathbf G^i\}_{i=1}^N}{\text{minimize}}
& &f(\mathbf \Phi x_0)\\
& \ \text{s.t.} &  &\begin{aligned} 
     & x_0 = x(t),\ \mathbf\Phi_x x_{0}\in\mathcal{X}, \ \mathbf\Phi_u x_{0}\in\mathcal{U},\\
     & [\mathbf \Phi^{i}]_{\textbf{in}_i(d)} = H_L([\tilde{\mathbf x}]_{\textbf{in}_i(d)},[\tilde{\mathbf u}]_{\textbf{in}_i(d+1)}) \mathbf G^i, \\ 
     &\mathbf G^i \text{ satisfies } \eqref{eqn:thm} \ \forall i=1,\dots,N.
\end{aligned}
\end{aligned}
\end{equation}

By introducing duplicate decision variables $\mathbf \Phi$ and $\mathbf G$, and rewriting the achievability and localization constraints in terms of the variable $\mathbf{G}$ by means of equation \eqref{eqn:thm}, the problem now enjoys a partially separable structure \citep{amoalonso_distributed_2020}. In what follows, we make such structure explicit and take advantage of it to distribute the problem across different subsystems via ADMM.

\subsection{A distributed subproblem solution via ADMM}

To rewrite \eqref{eqn:dd-DLMPC} and take advantage of the separability features in Assumption \ref{assump: locality}, we introduce the concept of row-wise separability and column-wise separability :
\begin{definition}
Given the partition $\{\mathfrak r_1,...,\mathfrak r_k\}$, a functional/set is \emph{row-wise separable} if:
\begin{itemize}
\item For a functional, $g(\mathbf \Phi) = \sum_{i = 1}^k g_i\big(\mathbf \Phi(\mathfrak{r}_i,:)\big)$ for some functionals $g_i$ for $i=1,...,k$.
\item For a set, $\mathbf \Phi \in \mathcal P$ if and only if $\mathbf \Phi(\mathfrak{r}_i,:) \in \mathcal P_i \ \forall i$ for some sets $\mathcal P_i$ for $i=1,...,k$.
\end{itemize}
\end{definition}
An analogous definition exists for \emph{column-wise separable} functionals and sets, where the partition $\{\mathfrak c_1,...,\mathfrak c_k\}$ entails the columns of $\mathbf \Phi$, i.e., $\mathbf \Phi(:,\mathfrak{c}_i)$. 

By Assumption \ref{assump: locality} the objective function and the safety/saturation constraints in equation \eqref{eqn:dd-DLMPC} are row-separable in terms of $\mathbf \Phi$. At the same time, the achievability and locality constraints \eqref{eqn:thm_Phi}, \eqref{eqn:thm} are column-separable in terms of $\mathbf G$. Hence, the data-driven MPC subroutine becomes:
\begin{equation} \label{eqn:dd-DLMPC1}
\begin{aligned}
& \underset{\mathbf \Phi, \mathbf \Psi, \{\mathbf G^i\}_{i=1}^N}{\text{minimize}}
& &\sum_{i=1}^{N}f^i([\mathbf \Phi]_i[x_{0}]_{\textbf{in}_i(d)})\\
& \ \text{s.t.} &  &\begin{aligned} 
     & x_0 = x(t),\ \mathbf G^i \text{ satisfies }\eqref{eqn:thm},\\
     &[\mathbf \Phi_x]_i [x_{0}]_{\textbf{in}_i(d)}\in\mathcal{X}^i, \ [\mathbf \Phi_u]_i [x_{0}]_{\textbf{in}_i(d)}\in\mathcal{U}^i, \quad \forall i=1,\dots,N,\\
     & [\mathbf \Psi^{i}]_{\textbf{in}_i(d)} = H_L([\tilde{\mathbf x}]_{\textbf{in}_i(d)},[\tilde{\mathbf u}]_{\textbf{in}_i(d+1)}) \mathbf G^i, \quad \forall i=1,\dots,N, \\
     & \mathbf \Phi = \mathbf \Psi.
\end{aligned}
\end{aligned}
\end{equation}

Notice that the objective function and the constraints decompose across the $d$-localized neighborhoods of each subsystem $i$. Given this structure, we can make use of ADMM to decompose this problem into row-wise local subproblems in terms of $\mathbf [\mathbf \Phi]_i$ and column-wise local subproblems in terms of $\mathbf G^i$ and $\mathbf \Psi^i$, both of which can also be parallelized across the subsystems. The ADMM subroutine iteratively updates the variables as
\begin{subequations}
\begin{align}
&\begin{aligned}
    \relax [\mathbf \Phi]_i^{\{k+1\}} &= \left\{
    \begin{aligned} \underset{[\mathbf \Phi]_i}{\text{argmin}} &\quad f^i([\mathbf \Phi]_i[x_{0}]_{\textbf{in}_i(d)}) + \frac{\rho}{2} \left\lVert g^i(\mathbf \Phi,\mathbf \Psi^{\{k\}},\mathbf \Lambda^{\{k\}}) \right\rVert_F^2 \\
    s.t. &\quad [\mathbf \Phi_x]_i [x_{0}]_{\textbf{in}_i(d)}\in\mathcal{X}^i, \\ &\quad [\mathbf \Phi_u]_i [x_{0}]_{\textbf{in}_i(d)}\in\mathcal{U}^i,\\ &\quad x_0 = x(t).
    \end{aligned} \right\}\label{eqn:update-Phi}
\end{aligned}\\
&\begin{aligned}
    \relax [\mathbf \Psi^{i}]_{\textbf{in}_i(d)}^{\{k+1\}} &= \left\{ \begin{aligned} \underset{[\mathbf \Psi^{i}]_{\textbf{in}_i(d)},\mathbf G^i}{\text{argmin}} &\quad \lVert g^{\textbf{in}_i(d)}\big([\mathbf \Phi^i]^{\{k+1\}}, [\mathbf \Psi^i],  [\mathbf \Lambda^i]^{\{k\}} \big)\rVert^2_F \\
    s.t. &\quad [\mathbf \Psi^{i}]_{\textbf{in}_i(d)} = [H_L(\tilde{\mathbf x},\tilde{\mathbf u})]_{\textbf{in}_i(d)}\mathbf G^i,\\ 
    &\quad \mathbf G^i\ \text{satisfies \eqref{eqn:thm}}.
    \end{aligned}  \right\}\label{eqn:update-Psi}
\end{aligned} \\
&[\mathbf \Lambda]_i^{\{k+1\}} = g^i( \mathbf \Phi^{\{k+1\}},\mathbf \Psi^{\{k+1\}},\mathbf \Lambda^{\{k\}})
\label{eqn:update-Lambda}
\end{align}
\end{subequations}
where we define 
$$g^*(\mathbf\Phi,\mathbf\Psi,\mathbf\Lambda):= [\mathbf \Phi]_* - [\mathbf \Psi]_* + [\mathbf \Lambda]_*$$
with $*$ denoting a subsystem or collection of subsystems, and
$$[H_L(\tilde{\mathbf x},\tilde{\mathbf u})]_{\textbf{in}_i(d)}:=H_L([\tilde{\mathbf x}]_{\textbf{in}_i(d)},[\tilde{\mathbf u}]_{\textbf{in}_i(d+1)}).$$

We note that to solve this subroutine, and in particular optimization \eqref{eqn:update-Psi}, each subsystem only needs to collect trajectory of states and control actions of subsystems that are at most $d+2$ hops away. This subroutine thus constitutes a distributed and localized solution. We also emphasize that the trajectory length only needs to scale with the $d$-localized system size instead of the global system size. The full algorithm is given in Algorithm \ref{alg:dd-DLMPC}.

\begin{algorithm}
\caption{Subsystem's $i$ implementation of D$^3$LMPC}\label{alg:dd-DLMPC}
\begin{algorithmic}[1]
\Statex \textbf{Input:} tolerance parameters $\epsilon_p,\epsilon_d >0$, Hankel matrices $[H_L(\tilde{\mathbf{x}}, \tilde{\mathbf{u}})]_{\mathbf{in}_i(d+1)}$ constructed from arbitrarily generated PE trajectories $\tilde{\mathbf{x}}_{\mathbf{in}_i(d+1)}, \tilde{\mathbf{u}}_{\mathbf{in}_i(d+1)}$.
\State Measure local state $[x_0]_{i}$, $k\leftarrow0$.
\State Share measurement $[x_0]_{i}$ with $\textbf{out}_{i}(d)$ and receive $[x_0]_{j}$ from $j\in\textbf{in}_{i}(d)$.
\State Solve optimization  problem (\ref{eqn:update-Phi}).
\State Share $[\mathbf \Phi]_{i}^{\{k+1\}}$ with $\textbf{out}_{i}(d)$. Receive the corresponding $[\mathbf \Phi]_{j}^{\{k+1\}}$ from $\textbf{in}_{i}(d)$ and construct $[\mathbf \Phi^i]_{\textbf{in}_i(d)}^{\{k+1\}}$.
\State Perform update (\ref{eqn:update-Psi}).
\State Share $[\mathbf \Psi^i]_{\textbf{in}_i(d)}^{\{k+1\}}$ with $\textbf{out}_{i}(d)$. Receive the corresponding $[\mathbf \Psi^j]_{\textbf{in}_i(d)}^{\{k+1\}}$ from $j\in\textbf{in}_{i}(d)$ and construct $[\mathbf \Psi]_{i}^{\{k+1\}}$.
\State Perform update (\ref{eqn:update-Lambda}).
\State \textbf{if}  $\left\Vert [\mathbf \Psi]_{i}^{\{k+1\}}-[\mathbf \Phi]_{i}^{\{k+1\}}\right\Vert_F\leq\epsilon_{p}$ and $\left\Vert [[\mathbf \Phi]_{i}^{\{k+1\}}-[\mathbf \Phi]_{i}^{\{k\}} \right\Vert_F\leq\epsilon_{d}$\textbf{:}
\Statex $\;\;$ Apply computed control action: $\,$ $[u_0]_i = \mathcal H_{L}([\tilde{\mathbf u}]_{\textbf{in}_i(d)}) \mathbf G^i [x_{0}]_{\textbf{in}_i(d)}$, and return to step 1.
\Statex \textbf{else:} 
\Statex $\;\;$ Set $k\leftarrow+1$ and return to step 3.
\end{algorithmic}
\end{algorithm}

\subsection{Theoretical Guarantees}

It is worth noting that this reformulation is equivalent to the closed-loop DLMPC also introduced in \citet{amoalonso_implementation_2021}, with the achievability constraint $Z_{AB}\mathbf \Phi = I$ replaced by the data-driven parametrization in terms of $\mathbf G$ and the Hankel matrix $H_L$. For this reason, guarantees derived for the DLMPC formulation \citep{amoalonso_theoretical_2021} directly apply to problem \eqref{eqn:dd-DLMPC} when the constraint sets $\mathcal X$ and $\mathcal U$ are polytopes.

\subsubsection{Convergence}

Algorithm \ref{alg:dd-DLMPC} relies on ADMM to separate both row, and column-wise computations, in terms of $\mathbf\Phi$ and $\mathbf G$ respectively. Each of these is then distributed into the subsystems in the network, and a communication protocol is established to ensure the ADMM steps are properly followed. Hence, one can guarantee the convergence of the data-driven version of DLMPC in the same way that convergence of model-based DLMPC is shown: by leveraging the convergence result of ADMM in [19]. For additional details see Lemma 2 in \citet{amoalonso_theoretical_2021}.

\subsubsection{Recursive Feasibility}

Recursive feasibility for formulation \eqref{eqn:DLMPC} is guaranteed by means of a localized maximally invariant terminal set $\mathcal X_T$. This set can be computed in a distributed manner and with local information only as described in \citet{amoalonso_theoretical_2021}. In particular, a closed-loop map $\mathbf \Phi$ for the unconstrained localized closed-loop system has to be computed. In the model-based SLS problem with quadratic cost, a solution exists for the infinite-horizon case \citep{yu_localized_2020}, which can be done in a distributed manner and with local information only. When no model is available, the same problem can be solved using the localized data-driven SLS approach introduced in \S IV with a finite-horizon approximation, which also allows for a distributed synthesis with only local data. The length of the time horizon chosen to solve the localized data-driven SLS problem might slightly impact the conservativeness of the terminal set, but since the conservatism in the FIR approach decays exponentially with the horizon length, this harm in performance is not expected to be substantial for usual values of the horizon. Once $\mathbf \Phi$ for the unconstrained localized closed-loop system has been computed, Algorithm 2 in \citet{amoalonso_theoretical_2021} can be used to synthesize this terminal set in a distributed and localized manner. Therefore, a terminal set that guarantees recursive feasibility for D$^3$LMPC can be computed in a distributed manner offline using only local information and without the need for a system model.

\subsubsection{Asymptotic Stability}
Similar to recursive feasibility, asymptotic stability for the D$^3$LMPC problem \eqref{eqn:dd-DLMPC} is directly inherited from the asymptotic stability guarantee for model-based DLMPC. In particular, adding a terminal cost based on the terminal set previously described is a sufficient condition to guarantee asymptotic stability of the DLMPC problem \citep{amoalonso_theoretical_2021}. Moreover, such cost can be incorporated in the D$^3$LMPC formulation in the same way as in the model-based DLMPC problem, and the structure of the resulting problem is analogous. This terminal cost introduces coupling among subsystems, but the coupling can be dealt with by solving step 3 of Algorithm \ref{alg:dd-DLMPC} via local consensus. Notice that since step 3 of Algorithm \ref{alg:dd-DLMPC} is written in terms of $\mathbf{\Phi}$, Algorithm 3 in \citet{amoalonso_theoretical_2021} can be directly used to handle this coupling. For additional details the reader is referred to \citet{amoalonso_theoretical_2021}.

\section{Simulation Experiments} \label{simulation}

We demonstrate through experiments that the D$^3$LMPC controller using only local data performs as well as a model-based DLMPC controller. We also show that for our algorithm, both runtime and the dimension of the data needed scale well with the size of the network. All the code needed to reproduce the experiments can be found at \url{https://github.com/unstable-zeros/dl-mpc-sls}.

\subsection{Setup} \label{sec:experiment-setup}
We evaluate the performance and scalablity of our algorithm on a system composed of a chain of subsystems, i.e., that $\mathcal E = \{ (i, i+1), (i+1, i),\ i=1,...,N-1\}$. Each subsystem $i$ has a 2-dimensional state $[\theta_i,\omega_i]^\intercal$, where $\theta_i$ and $\omega_i$ are respectively the phase angle deviation and frequency deviation of the subsystem. We assume that each subsystem takes a scalar control action $u_i$. We consider the same type of dynamic coupling between the subsystem as that in  \citet{amoalonso_implementation_2021}, which is given as
\[
\begin{bmatrix} \theta(t+1) \\ \omega(t+1) \end{bmatrix}_i = \sum_{j \in \textbf{in}_i(1)} [A]_{ij}
\begin{bmatrix} \theta(t) \\ \omega(t) \end{bmatrix}_i + [B]_i [u]_i,
\]
where
\[
[A]_{ii} = \begin{bmatrix} 1 & \Delta t \\ -\frac{k_i}{m_i}\Delta t & 1-\frac{d_i}{m_i}\Delta t \end{bmatrix}, \quad [A]_{ij} = \begin{bmatrix} 0 & 0 \\ \frac{k_{ij}}{m_i}\Delta t & 0\end{bmatrix},
\]
and $B_{ii} = [0\;1]^\top$ for all $i$. The parameters $m_i, d_i, k_{ij}$ are sampled uniformly at random from the intervals $[0, 2], [0.5, 1], [1, 1.5]$, respectively and $k_i := \sum_{j \in \textbf{in}_i(1)}k_{ij}$. Finally, the discretization time is set to be $\Delta t = 0.2$. The goal of the controller is to minimize the LQR cost with $Q=I$, $R=I$.

To show the optimality of our approach in Section~\ref{sec:experiment-optimality}, we consider a base system with 64 subsystems. To demonstrate the scalability of our method, we consider systems of varying sizes for our experiment in Section~\ref{sec:experiment-scalability}. Unless mentioned otherwise, for all the experiments, we consider a locality region of size $d=2$, use a planning horizon of $T=5$ steps and simulate the system forward for $30$ steps.

\begin{figure*}[ht]
    \centering
    \includegraphics[width=\textwidth]{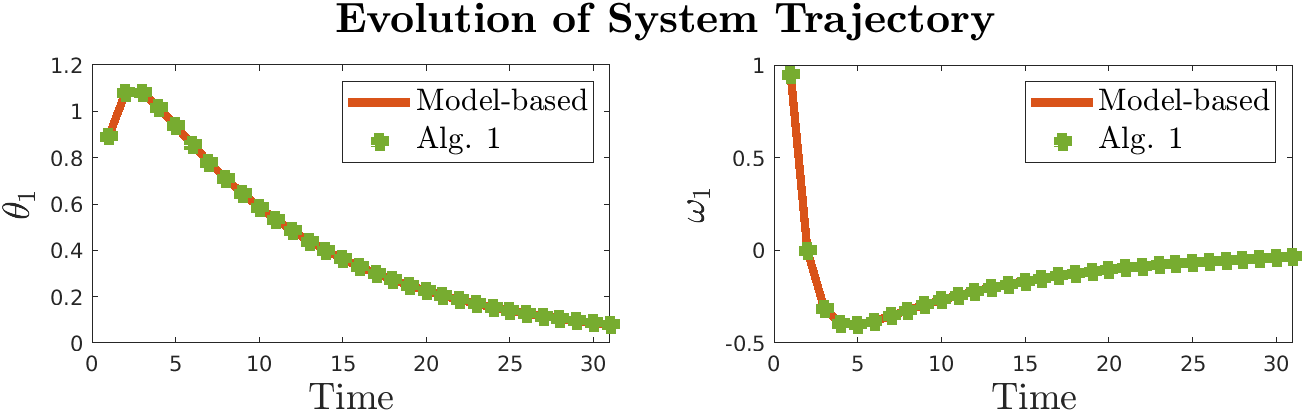}
    \caption{The trajectory generated by the model-based DLMPC algorithm (solid orange line) vs. the trajectory generated by D$^3$LMPC (green circles). We observe that the two coincides.}
    \label{fig:optimality-chain}
\end{figure*}

\subsection{Optimal performance} \label{sec:experiment-optimality}
We evaluate the performance of D$^3$LMPC (Alg.~\ref{alg:dd-DLMPC}) on the system described in Section~\ref{sec:experiment-setup}. First, we show that the trajectory given by the D$^3$LMPC algorithm matches the trajectory generated by a model-based DLMPC solved with a centralized solver \citep{cvx,gurobi}. Notice that the model-based DLMPC solves the optimization problem \eqref{eqn:DLMPC} with perfect knowledge of the system dynamics, while the D$^3$LMPC  algorithm (Alg.~\ref{alg:dd-DLMPC}) only has access to local past trajectories. Due to space constraints, we only show the state trajectory of the first subsystem (Figure~\ref{fig:optimality-chain}). We observe that the trajectory generated by our controller matches the trajectory of the optimal controller with the same locality region size. Further, the optimal cost for both schemes is the same up to numerical precision. This confirms that the D$^3$LMPC algorithm (Alg.~\ref{alg:dd-DLMPC}) can synthesize optimal controllers using only local trajectory data and no knowledge of the system dynamics.

\begin{figure}[h]
    \centering
    \includegraphics[width=0.6\textwidth]{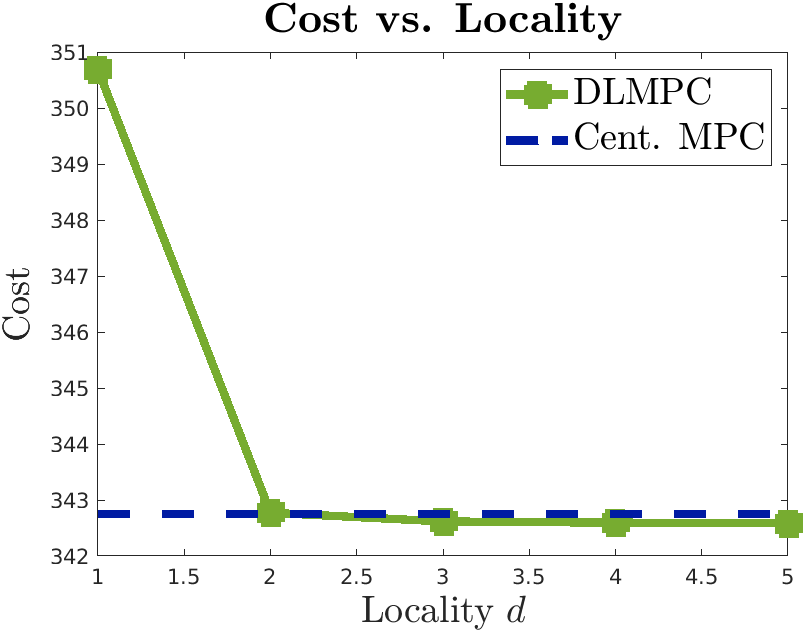}
    \caption{(Green) The cost achieved by the optimal controller versus the size of the locality region for the system response. (Blue) The cost achieved by a centralized MPC controller, i.e., that it has no locality constraints. We observe that cost for the distributed controllers decreases as the size of the locality region grows and approaches the cost of the centralized controller.}
    \label{fig:cost-locality-chain}
\end{figure}
\begin{figure}[h]
    \centering
    \includegraphics[width=0.6\textwidth]{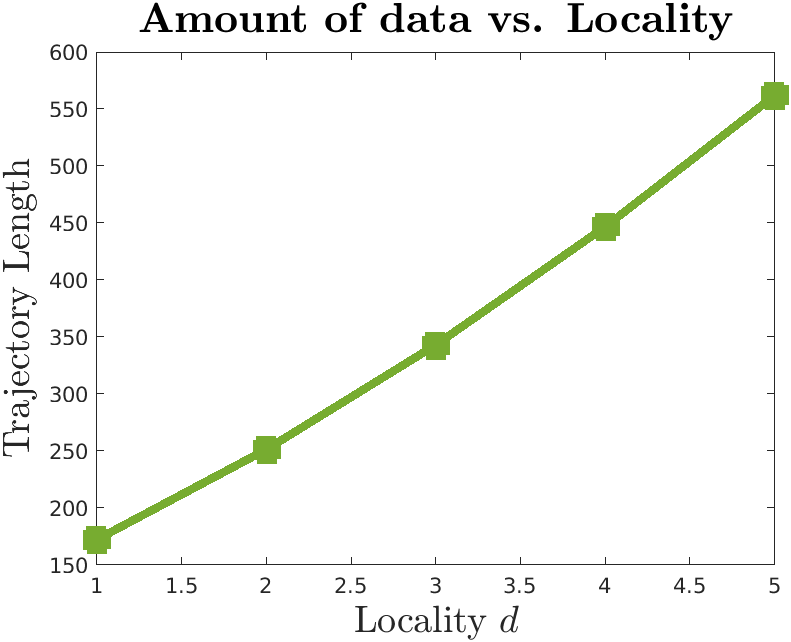}
    \caption{The growth of necessary length of collected trajectory versus the size of the locality region of system response. The trajectory length grows with the size of the locality region.}
    \label{fig:tl-d-chain}
\end{figure}

We further highlight the relevance of locality region size on the optimality of the solution. The size of the locality region $d$ can be seen as a design parameter in Alg.~\ref{alg:dd-DLMPC} that allows one to tradeoff between computation complexity and performance of the controller. In Figure~\ref{fig:cost-locality-chain}, we show how the optimal cost varies with the size of the locality region on the same system. As the size of the locality region grows, the optimal cost decreases. This matches the intuition that by allowing each subsystem to influence more subsystems, and as more information is made available to each subsystem, controllers of better quality can be synthesized. We note that the performance improvement by increasing the locality region size is the most significant when the locality region is small. In Figure~\ref{fig:tl-d-chain} we simultaneously show how much the trajectory length needs to grow with the size of the locality region to satisfy the persistence of excitation condition for applying Willem's Fundamental Lemma. We note that the growth in the necessary trajectory length not only means longer trajectory needs to be collected, but also means that the size of the optimization problem grows, thus incurring higher computation complexity for each optimization step.  Hence, the choice of an optimal $d$ heavily depends on the specific application considered.

\subsection{Scalability} \label{sec:experiment-scalability}
First, we show that the runtime of our method scales well with the size of the global system. We consider systems composed of $9,\ 16,\ 36,\ 64,\ 81,\ 100,$ and $121$ subsystems. For each system size, we randomly generate $10$ different systems and report the average computation time per MPC step.\footnote{Runtime is measured after the first iteration to compute the runtime of the MPC algorithm after warmstart. The optimization problems were solved using the Gurobi \citep{gurobi} optimizer on a personal desktop computer with an 8-core Intel i7 processor.} The result is shown in Figure~\ref{fig:scalability-chain}. We note that the runtime only increased $2\times$ while the size of the system has increased more than $12\times$. Further, the growth of the runtime flattens as the size of the network grows, suggesting that our method scales well on sparsely connected systems. This trend has previously been observed with ADMM schemes for MPC \citep{conte_computational_2012}.
\begin{figure}[H]
    \centering
    \includegraphics[width=0.6\textwidth]{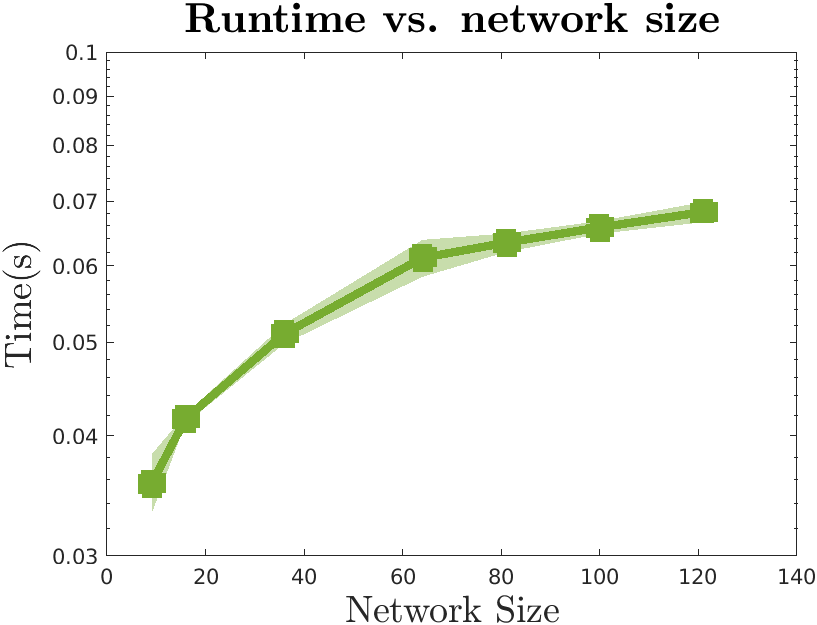}
    \caption{The average per-step per-subsystem runtime of the MPC algorithm. The solid line shows the average runtime over 10 randomly generated systems, and the circles represent the runtime for each of the 10 randomly generated instances for each system size.}
    \label{fig:scalability-chain}
\end{figure}

\begin{figure}[H]
    \centering
    \includegraphics[width=0.6\textwidth]{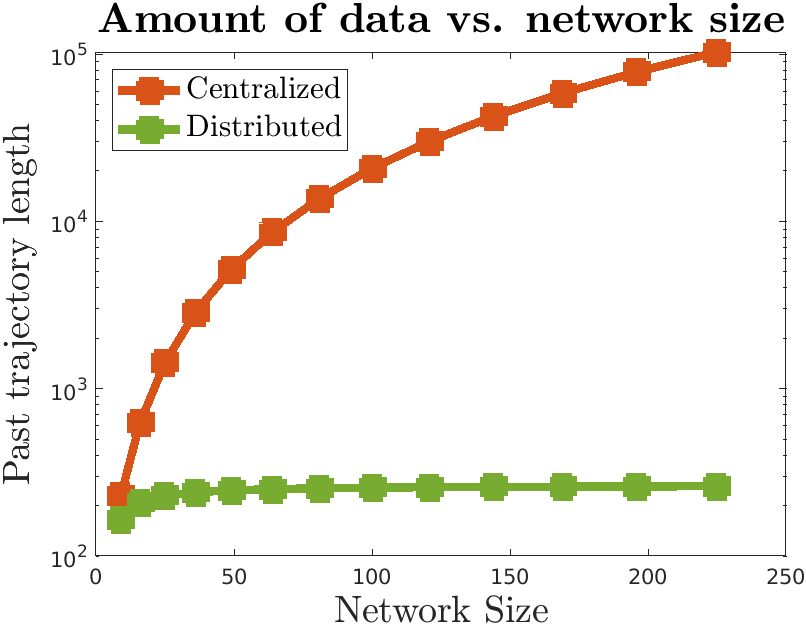}
    \caption{Length of necessary trajectory length versus network size. Note that this is plotted on a semilog axis. Our distributed approach requires much shorter trajectory over a centralized data-driven approach.}
    \label{fig:tl-network-chain}
\end{figure}

Next, we show that the length of the trajectory that needs to be collected for the D$^3$LMPC controller grows much more slowly than that for a centralized data-driven method that does not exploit the locality structure of the problem (equivalent to solving the SLS problem with constraints \eqref{eqn:localized_G} instead of \eqref{eqn:thm}). The result is shown in Figure~\ref{fig:tl-network-chain}. We note that our method requires less data (length of the trajectory) to be collected in general. At the same time, the larger the system, the more benefit one gets from using our method over a centralized data-driven approach. 

\section{Conclusion}

In this paper we define and analyze a data-driven Distributed and Localized Model Predictive Control (D$^3$LMPC) scheme. This approach can synthesize optimal localized control policies using only local communication and requires no knowledge of the system model. We base our results on the data-driven SLS approach \citep{xue_data_2021}, and extend this framework to allow for locality constraints. We then use these results to provide an alternative data-driven synthesis for the DLMPC algorithm \citep{amoalonso_distributed_2020} by exploiting the separability of the problem via ADMM. The resulting algorithm enjoys the same scalability properties as model-based DLMPC \citep{amoalonso_implementation_2021} and only need trajectory data that scales with the size of the $d$-localized neighborhood. Moreover, recursive feasibility and stability guarantees that exist for model-based DLMPC \citep{amoalonso_theoretical_2021} directly apply to this framework.

The work presented here is, to the best of our knowledge, the first fully distributed and localized data-driven MPC approach that achieves globally optimal performance with local information collection and communication among subsystems. This, when extended to the noisy settings, offers a promising avenue forward towards localized and scalable learning and control with guarantees.




\bibliographystyle{unsrtnat}
\bibliography{bibliography/references}

\end{document}